\documentclass[final,12pt]{elsarticle}
\usepackage{amsthm}
\usepackage{amsmath}
\usepackage{upref}
\usepackage{amsfonts}
\usepackage{amssymb}
\usepackage{mathpazo}
\usepackage{times}

\newcommand{\cC}{\mathcal{C}}

\newcommand{\cE}{\mathcal{E}}

\newcommand{\cT}{\mathcal{T}}

\newcommand{\mathset}[1]{\left\{#1\right\}}
\newcommand{\abs}[1]{\left|#1\right|}

\newcommand{\floorenv}[1]{\left\lfloor #1 \right\rfloor}
\newcommand{\parenv}[1]{\left( #1 \right)}

\newcommand{\be}[1]{\begin{equation}\label{#1}}
\newcommand{\ee}{\end{equation}}

\renewcommand{\leq}{\leqslant}

\renewcommand{\geq}{\geqslant}

\renewcommand{\Bbb}{\mathbb}

\newcommand{\Cref}[1]{Co\-ro\-lla\-ry\,\ref{#1}}

\renewcommand{\Bbb}{\mathbb}

\newcommand{\C}{{\Bbb C}}

\newcommand{\R}{{\Bbb R}}
\newcommand{\Z}{{\Bbb Z}}
\newcommand{\Q}{{\Bbb Q}}

\newcommand{\km}{k_{-}}
\newcommand{\kp}{k_{+}}

\newcommand{\leg}[2]{\genfrac{(}{)}{}{}{#1}{#2}}
\newcommand{\qr}{\text{QR}}
\newcommand{\nr}{\text{QNR}}
\newcommand{\ove}{\overline{e}}
\newcommand{\ovx}{\overline{x}}
\newcommand{\ovy}{\overline{y}}
\newcommand{\ovt}{\overline{t}}

\newtheorem{theorem}{Theorem}
\newtheorem{lemma}[theorem]{Lemma}
\newtheorem{corollary}[theorem]{Corollary}
\newdefinition{definition}{Definition}
\newdefinition{remark}{Remark}
\newdefinition{const}{Construction}
\newdefinition{example}{Example}

\begin{document}

\journal{Journal of Combinatorial Theory Ser.~A}

\begin{frontmatter}

\title{On the Non-existence of Lattice Tilings by Quasi-crosses\tnoteref{lbl1}}
\tnotetext[lbl1]{This work was supported in part by ISF grant 134/10.}

\author{Moshe Schwartz\fnref{lbl2}}
\ead{schwartz@ee.bgu.ac.il}
\address{Department of Electrical and Computer Engineering, Ben-Gurion
University of the Negev, Israel} 
\fntext[lbl2]{The author is on a sabbatical at the Department of Electrical
Engineering, MIT, Research Laboratory of Electronics.}


\begin{abstract}
We study necessary conditions for the existence of lattice tilings of
$\R^n$ by quasi-crosses. We prove non-existence results, and focus in
particular on the two smallest unclassified shapes, the
$(3,1,n)$-quasi-cross and the $(3,2,n)$-quasi-cross. We show that for
dimensions $n\leq 250$, apart from the known constructions, there are
no lattice tilings of $\R^n$ by $(3,1,n)$-quasi-crosses except for ten
remaining cases, and no lattice tilings of $\R^n$ by
$(3,2,n)$-quasi-crosses except for eleven remaining cases.
\end{abstract}

\begin{keyword}
tiling \sep lattices \sep quasi-cross \sep group splitting
\MSC[2010] 05B45 \sep 52C22
\end{keyword}

\end{frontmatter}


\section{Introduction}
\label{sec:introduction}

Problems involving tilings of $\R^n$ by clusters of cubes have a long
history, as is evident from the early work of Minkowski \cite{Min07}.
In this context, let
\[Q=\mathset{(x_1,\dots,x_n) ~|~ 0\leq x_i < 1, x_i\in\R}\]
denote the \emph{unit cube}, which, we shall also say, is
\emph{centered at the origin}. A \emph{translate} of the cube
by a vector $\ove\in\R^n$ is the set
\[\ove+Q=\mathset{\ove+\ovx ~|~ \ovx\in Q},\]
and a \emph{cluster of cubes} is a union of
disjoint translates of cubes
\[\cC=\cE+Q=\mathset{\ove+Q ~|~ \ove\in\cE},\]
for some $\cE\subseteq \R^n$.

A set of disjoint translates of $\cC$ is called a \emph{packing}
of $\R^n$ by $\cC$. If the union of the translates is the entire space
$\R^n$, we say it is a \emph{tiling}. If the set of translates forming
the packing (tiling) forms an additive subgroup of $\Z^n$, we shall say
it is a lattice\footnote{
This is, in fact, an \emph{integer} lattice, but we shall omit
this throughout the paper.} packing (lattice tiling).

Several types of clusters have been considered in the past. The two
most studied clusters are the \emph{$(k,n)$-cross} and the
\emph{$(k,n)$-semi-cross}.  The $(k,n)$-cross is defined by the
following set of translates
\[\cE_{\mathrm{cross}}=\mathset{i\ove_j\in\R^n ~|~ i\in[-k,k], j\in[n]}\]
where $[a,b]=\mathset{a,a+1,\dots,b}\subseteq\Z$, $[a]$ is short for
$[1,a]$, and $\ove_j$ is the $j$-th standard unit vector. That is, a
$(k,n)$-cross contains a center cube, and arms of length $k$ cubes in
the positive and negative directions along each axis.  In contrast,
the $(k,n)$-semi-cross has arms only in the positive direction and is
defined by
\[\cE_{\mathrm{semi-cross}}=\mathset{i\ove_j\in\R^n ~|~ i\in[0,k], j\in[n]}\]

Packings (lattice and non-lattice) of $\R^n$ by crosses and
semi-crosses were studied by Stein \cite{Ste84}, and Hickerson and
Stein \cite{HicSte86}. For an excellent survey the reader is referred
to \cite{SteSza94}. We also note that a $(1,n)$-cross is also a Lee
sphere of radius $1$. Apart from radius $1$ or dimension $2$, the
non-existence of tilings of $\R^n$ by Lee spheres is a long-standing
conjecture by Golomb and Welch \cite{GolWel70} (see
\cite{Pos75,Hor09a,Hor09b,Etz11}, as well as the more recent
\cite{HorAlb12} for a survey on the current status of the conjecture).

Motivated by an application to error-correcting codes for non-volatile
memories, Schwartz \cite{Sch12} suggested a generalization of both the
cross and semi-cross to a shape called the
\emph{$(\kp,\km,n)$-quasi-cross} defined by the set of translations
\[\cE_{\mathrm{quasi-cross}}=\mathset{i\ove_j\in\R^n ~|~ i\in[-\km,\kp], j\in[n]}.\]
Namely, in a $(\kp,\km,n)$-quasi-cross the center cube has arms of
length $\kp$ in the positive direction, and arms of length $\km$ in
the negative direction (see Figure \ref{fig:exampleqc}). Thus,
a $(k,0,n)$-quasi-cross is simply a $(k,n)$-semi-cross, while a 
$(k,k,n)$-quasi-cross is a $(k,n)$-cross. To avoid the two studied cases
we shall assume throughout that $1\leq \km < \kp$.

\begin{figure}[ht]
\begin{center}
\includegraphics[scale=0.7]{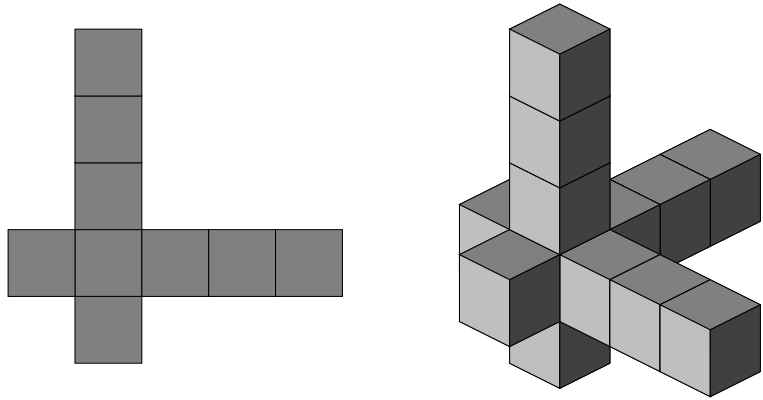}
\end{center}
\caption{A $(3,1,2)$-quasi-cross and a $(3,1,3)$-quasi-cross}
\label{fig:exampleqc}
\end{figure}

A few constructions were given in \cite{Sch12} for lattice tilings of
$\R^n$ by quasi-crosses, and in particular, a full classification was
provided of the dimensions in which there exist lattice tilings by
$(2,1,n)$-quasi-crosses. Recently, Yari et al.~\cite{YarKloBos13} gave
other constructions for lattice packings and tilings by quasi-crosses,
and in particular, new constructions for tilings by
$(3,1,n)$-quasi-crosses.

The motivation given in \cite{Sch12} is that of producing perfect
$1$-error-correcting codes for the unbalanced limited magnitude
channel, a natural extension to the earlier work of
\cite{CasSchBohBru10}. The dual case of $(n-1)$-error-correcting codes
gives rise to a tiling problem of cluster of cubes called a ``chair'',
which is described in \cite{BuzEtz12}.

The goal of this work is to derive new necessary conditions for the
existence of tilings of $\R^n$ by quasi-crosses. Though most of the
results apply to general $(\kp,\km,n)$-quasi-crosses, we shall focus
in particular on the two smallest unclassified cases of the
$(3,1,n)$-quasi-cross and the $(3,2,n)$-quasi-cross.

The paper is organized as follows: We begin in Section
\ref{sec:prelim} by providing the notation and definitions used
throughout the paper.  We shall also cite relevant results from
previous works. We continue in Section \ref{sec:main} with a list of
the main results. We conclude in Section \ref{sec:conc} with the
application of the main results to the specific case of tilings
by $(3,1,n)$-quasi-crosses and tilings by $(3,2,n)$-quasi-crosses.


\section{Preliminaries}
\label{sec:prelim}

We shall now describe the definitions and notation used in this work.
For the reader's benefit we repeat some of the definitions given in
the introduction. A \emph{cube} is defined as the set
\[Q=\mathset{(x_1,\dots,x_n) ~|~ 0\leq x_i < 1, x_i\in\R}.\]
A set of pair-wise disjoint translates of the cube is a \emph{cluster
  of cubes}
\[\cC=\cE+Q=\mathset{\ove+Q ~|~ \ove\in\cE},\]
for some $\cE\subseteq \R^n$ specifying the translate vectors. Through
the paper we shall use only integer translate vectors, i.e.,
$\cE\in\Z^n$.

We denote $[a,b]=\mathset{a,a+1,\dots,b}\subseteq\Z$, $[a]$ is short
for $[1,a]$, and $[a,b]^*=[a,b]\setminus\mathset{0}$. For any two
positive integers $1\leq \km<\kp$, the $(\kp,\km,n)$-quasi-cross is
the cluster of cubes defined by the translate vectors
\[\cE_{\mathrm{quasi-cross}}=\mathset{i\ove_j\in\R^n ~|~ i\in[-\km,\kp], j\in[n]}.\]

Let $\cT\subseteq\R^n$ be a set of vectors, and let
$\cC_{\mathrm{quasi-cross}}$ be a $(\kp,\km,n)$-quasi-cross cluster of
cubes centered at the origin. If the translates
$\ovt+\cC_{\mathrm{quasi-cross}}$, $\ovt\in\cT$, are pair-wise
disjoint, we say $\cT$ is a \emph{packing} of $\R^n$ by
$(\kp,\km,n)$-quasi-crosses. If
\[\bigcup_{\ovt\in\cT}\parenv{\ovt+\cC_{\mathrm{quasi-cross}}}=\R^n\]
we say $\cT$ is a \emph{tiling} of $\R^n$ by
$(\kp,\km,n)$-quasi-crosses.  If $\cT$ is an additive subgroup of
$\Z^n$ then we shall call $\cT$ a lattice, and will use the letter
$\Lambda$ instead of $\cT$ to denote it.

Finally, the \emph{primorial} is defined as
\[n\#=\prod_{\substack{\text{$p$ prime}\\ p\leq n}} p.\]

\subsection{Abelian-Group Splitting and Lattice Tiling}

While we may use geometric arguments to prove necessary conditions for
a shape to tile $\R^n$, stronger results may be obtained using the
algebraic structure of a \emph{lattice} tiling. An equivalence between
lattice tilings and Abelian-group splitting was described in
\cite{Ste67,Ste84,HicSte86}, which we describe here for
completeness.

Let $G$ be an finite Abelian group, where we shall denote the group
operation as $+$. Given some $s\in G$ and a non-negative integer
$m\in\Z$, we denote by $ms$ the sum $s+s+\dots+s$, where $s$ appears
in the sum $m$ times. The definition is extended in the natural way to
negative integers $m$.

A splitting of $G$ is a pair of sets,
$M\subseteq\Z\setminus\mathset{0}$, called the \emph{multiplier set},
and $S=\mathset{s_1,s_2,\dots,s_n}\subseteq G$, called the
\emph{splitter set}, such that the elements of the form $ms$, $m\in
M$, $s\in S$, are all distinct, non-zero, and cover all the non-zero
elements in $G$. We shall denote such a splitting as $G=(M,S)$. It
follows that $\abs{M}\cdot\abs{S}=\abs{G}-1$.

Next, we define a homomorphism $\phi:\Z^n\rightarrow G$ by
\[\phi(x_1,x_2,\dots,x_n)=\sum_{i=1}^{n} x_i s_i.\]
If the multiplier set is $M=[-\km,\kp]^*$, then it may be easily
verified that $\ker \phi$ is a lattice tiling of $\R^n$ by
$(\kp,\km,n)$-quasi-crosses. The fact that $\ker \phi$ is a lattice is
obvious. To show that the lattice is a packing by
$(\kp,\km,n)$-quasi-crosses, assume to the contrary two such distinct
quasi-crosses, one centered at $\ovx=(x_1,\dots,x_n)$ and one centered
at $\ovy=(y_1,\dots,y_n)$, have a non-empty intersection, i.e.,
\[\ovx+m_1\ove_i=\ovy+m_2 \ove_j,\]
where $m_1,m_2\in M$, then
\[m_1 s_i=\phi(\ovx+m_1 \ove_i)=\phi(\ovy+m_2 \ove_j)= m_2 s_j\]
which is possible only if $m_1=m_2$ and $i=j$, resulting in the two
quasi-crosses being the same one, a contradiction.

Finally, to show that the packing is in fact a tiling let
$\ovx\in\R^n$ be some point in the space. Obviously, $\ovx\in
\floorenv{\ovx}+Q$. If $\phi(\floorenv{\ovx})=0$ then
$\floorenv{\ovx}\in\ker\phi$ and $\ovx$ is in the
$(\kp,\km,n)$-quasi-cross cube cluster centered at
$\floorenv{\ovx}$. Otherwise, by the properties of the splitting there
exist $m\in M$ and $s_i\in S$ such that
$\phi(\floorenv{\ovx})=ms_i$. It follows that
$\floorenv{\ovx}-m\ove_i\in\ker\phi$ and $\ovx$ is in the
$(\kp,\km,n)$-quasi-cross cube cluster centered at
$\floorenv{\ovx}-m\ove_i$.

Group splitting as a method for constructing error-correcting codes
was also discussed, for example, in the case of shift-correcting codes
\cite{Tam98} and integer codes \cite{Tam05}.

\subsection{Previous Results}

Several results from previous works are relevant to this one. Some
apply directly to quasi-crosses, while others will be used as a basis
for our new results, appearing in the next section. The first theorem
we cite is the only one which uses geometric arguments to derive a
necessary condition on lattice tilings of $\R^n$ by quasi-crosses.

\begin{theorem}\cite[Theorem 9]{Sch12}
\label{th:geometry}
For any $n\geq 2$, if
\[\frac{2\kp(\km+1)-\km^2}{\kp+\km}>n,\]
then there is no lattice tiling of $\R^n$ by
$(\kp,\km,n)$-quasi-crosses.
\end{theorem}

When looking for a lattice tiling using the group splitting
equivalence, the question is which finite Abelian group to split,
where it was demonstrated in \cite{Sch12} that splitting different
Abelian groups of the same size may result in different lattice
tilings. However, since we are only interested in finding necessary
conditions for the existence of lattice tilings, the following theorem
from \cite{Sch12} (which is a generalization of a theorem from
\cite{SteSza94}) shows that we may focus only on cyclic groups.

\begin{theorem}\cite[Theorem 15]{Sch12}
\label{th:cyclic}
Let $G$ be a finite Abelian group, and let $M=[-\km,\kp]^*$ be the
multiplier set corresponding to the $(\kp,\km,n)$-quasi-cross. If
there is a splitting $G=(M,S)$, then there is a splitting of the
cyclic group of the same size $\Z_{\abs{G}}=(M,S')$.
\end{theorem}

Using Theorem \ref{th:cyclic} we can say that the
$(\kp,\km,n)$-quasi-cross lattice tiles $\R^n$ if and only if
$\Z_q=(M,S)$, where $q=n(\kp+\km)+1$ and
$M=[-\km,\kp]^*$. Furthermore, the expressions $ms$, for $m\in M$ and
$s\in S$, simply denote integer multiplication in the ring $\Z_q$.  To
avoid confusion, we shall denote the multiplicative semi-group of the
ring $\Z_q$ as $R_q$.

Another result which shall be useful for the classification of lattice
tilings by $(3,2,n)$-quasi-crosses is the following.

\begin{theorem}\cite[Theorem 16]{Sch12}
\label{th:kmo}
Let $k\geq 2$ be some positive integer, and let $M=[-(k-1),k]^*$. If
$G=(M,S)$ is a splitting of an Abelian group $G$, $\abs{G}>1$, then
$\gcd(k,\abs{G})\neq 1$.
\end{theorem}

A notion we shall find useful is that of a \emph{character}, as defined
by Stein \cite{Ste67}: A character is a homomorphism $\chi:G\to R$
from a semi-group $G$ into a (multiplicative) semi-group $H$. The
following theorem, with a one-line proof that we bring for completeness,
is due to Stein\footnote{The version due to Stein is somewhat more
general, but we shall not require the full generality of the original
claim.}.

\begin{theorem}\cite[Theorem 4.1]{Ste67}
\label{th:charsum}
Let us consider a splitting $\Z_q=(M,S)$ and let $\chi:R_q\to R$ be a character
from $R_q$ into a ring $R$. Then
\[\parenv{\sum_{m\in M}\chi(m)}\cdot\parenv{\sum_{s\in S}\chi(s)}
=\sum_{a\in R_q}\chi(a).\]
\end{theorem}
\begin{proof}
\[\sum_{a\in R_q}\chi(a)=\sum_{\substack{m\in M \\s\in S}}\chi(ms)=
\sum_{\substack{m\in M\\s\in S}}\chi(m)\chi(s)=\parenv{\sum_{m\in M}\chi(m)}\cdot\parenv{\sum_{s\in S}\chi(s)}.\]
\end{proof}

Several characters will be of interest in the following section, the first
is the Legendre symbol: for an odd prime $p$ we define the character
$\leg{\cdot}{p}:R_p\to \C$ as
\[\leg{a}{p}=\begin{cases}
1 & \text{$a\equiv x^2\pmod{p}$ for some $x\in R_p$,} \\
-1 & \text{otherwise.}
\end{cases}
\]
If $\leg{a}{p}=1$ we call $a$ a \emph{quadratic residue modulo $p$}
(QR), and otherwise we call $a$ a \emph{quadratic non-residue modulo
$p$} (QNR). Using the Legendre symbol and Theorem \ref{th:charsum}
Stein proved the following theorem.

\begin{theorem}\cite[Corollary 4.3]{Ste67}
\label{th:stein43}
If $\Z_p=(M,S)$ is a splitting, $p$ an odd prime,
then in at least one of $M$ or $S$ the number of quadratic residues equals
the number of quadratic non-residues.
\end{theorem}

We recall some well-known facts about the Legendre symbol, which we
shall use later. For a proof, see for example \cite{HarWri08}.
\begin{lemma}
\label{lem:qr}
Let $p$ be an odd prime, and let $\ell$ denote some integer. Then
\[
\begin{array}{cl}
\leg{-1}{p}=1 & \text{iff $p=4\ell+1$,}\\
\leg{2}{p}=1 & \text{iff $p=8\ell\pm 1$,}\\
\leg{3}{p}=1 & \text{iff $p=12\ell\pm 1$,}\\
\leg{5}{p}=1 & \text{iff $p=10\ell\pm 1$.}
\end{array}
\]
\end{lemma}


\section{Main Results}
\label{sec:main}

In this section we list our new results, where we group them according
to the method employed to derive the necessary condition for lattice
tiling $\R^n$ by $(\kp,\km,n)$-quasi-crosses.

\subsection{The Legendre Symbol and Higher-Order Characters}

We begin our treatment by examining results obtained by using the Legendre
symbol and higher-order characters.

\begin{theorem}
\label{th:3mod6}
The $(3,1,n)$-quasi-cross does not lattice tile $\R^n$ when $4n+1$ is
a prime, $n\equiv 3\pmod{6}$.
\end{theorem}
\begin{proof}
Assume to the contrary that $\Z_{4n+1}=(M,S)$ is a splitting with
$M=[-1,3]^*$. Obviously, $1$ is a QR. Using Lemma
\ref{lem:qr} we note that $-1$ and $3$ are also a QRs, while $2$ is a
QNR, when $n\equiv 3\pmod{6}$. Thus $S$ should have an equal number of
QRs and QNRs, but $\abs{S}=n$ is odd, a contradiction.
\end{proof}

Up to dimension $250$ Theorem \ref{th:3mod6} rules out lattice tilings by
$(3,1,n)$-quasi-crosses for
\begin{align*}
n&=3,9,15,27,39,45,57,69,87,93,99,105,135,153,165,177,183,189,\\
&\quad\  207,213,219,249.
\end{align*}
Equally simple, but more tedious, the same method applies for larger
quasi-crosses.

\begin{theorem}
The $(5,1,n)$-quasi-cross does not lattice tile $\R^n$ when $6n+1$ is
a prime, $n\equiv 5,7,11,13,19\pmod{20}$.
\end{theorem}
\begin{proof}
Assume to the contrary that $\Z_{6n+1}=(M,S)$ is a splitting with
$M=[-1,5]^*$. Denote $n=20\ell+r$, with
$r\in\mathset{5,7,11,13,19}$. The following table summarizes which of
the elements of $M$ is a QR using Lemma \ref{lem:qr}:
\[
\begin{array}{c||c|c|c|c|c|c}
6n+1 & -1 & 1 & 2 & 3 & 4 & 5\\
\hline
120\ell+31  & \nr & \qr & \qr & \nr & \qr & \qr \\
120\ell+43  & \nr & \qr & \nr & \nr & \qr & \nr \\
120\ell+67  & \nr & \qr & \nr & \nr & \qr & \nr \\
120\ell+79  & \nr & \qr & \qr & \nr & \qr & \qr \\
120\ell+115 & \nr & \qr & \nr & \nr & \qr & \nr
\end{array}
\]
We note that in all cases, $M$ does not contain an equal number of QRs and
QNRs. Thus, $S$ must contain an equal number of QRs and QNRs and so
$\abs{S}=n$ must be even, a contradiction.
\end{proof}

A generalization for higher power residues (generalizing the Legendre
symbol) can be made, as is seen in the next theorem, which uses
quartic residues\footnote{Quartic residues are sometimes also called
  \emph{biquadratic} residues.}.

\begin{theorem}
\label{th:char4}
Let $4n+1$ be a prime, with $n$ being an odd integer. If
\[6^n \not\equiv 1 \pmod{4n+1}\]
then the $(3,1,n)$-quasi-cross does not lattice tile $\R^n$.
\end{theorem}
\begin{proof}
Since $4n+1$ is a prime, $\Z_{4n+1}$ is a field, and so let $g$ be a
primitive element in $\Z_{4n+1}$. We define the character
$\chi:R_{4n+1}\rightarrow \C$ as 
\[\chi(g^j)=e^{\frac{2\pi ij}{4}},\]
where $i=\sqrt{-1}$.

Assume to the contrary that there exists a splitting $\Z_{4n+1}=(M,S)$
under the conditions of the theorem.
By Theorem \ref{th:charsum} we have
\begin{equation}
\label{eq:4}
\parenv{\sum_{m\in M}\chi(m)}\cdot\parenv{\sum_{s\in S}\chi(s)}=\sum_{j=1}^{4n}\chi(j).
\end{equation}
We also have
\begin{equation}
\label{eq:5}
\sum_{j=1}^{4n}\chi(j)=\sum_{j=0}^{4n-1}\chi(g^j)=
\sum_{j=0}^{4n-1}\chi(g)^j=\frac{\chi(g)^{4n}-1}{\chi(g)-1}=0.
\end{equation}
If follows from \eqref{eq:4} and \eqref{eq:5} that
\[\sum_{m\in M}\chi(m)=0 \qquad \text{or} \qquad \sum_{s\in S}\chi(s)=0.\]

We first note that $1$ and $-1$ are quadratic residues in
$\Z_{4n+1}$. If $2$ or $3$ are quadratic residues, then by Lemma
\ref{lem:qr} the set $S$ must contain an equal number of quadratic
residues and quadratic non-residues, but $\abs{S}=n$ is odd. We
therefore need to consider only the case where both $2$ and $3$ are
quadratic non-residues.

We now turn to check the characters of the elements of $M$. It is easily
seen that $\chi(1)=1$. Since $n$ is odd, we deduce $\chi(-1)=-1$, i.e.,
$-1$ is a quadratic residue in $\Z_{4n+1}$ but is a quartic non-residue.
Since both $2$ and $3$ are quadratic non-residues, we have
$\chi(2),\chi(3)\in\mathset{i,-i}$.

We note that the quartic residues form a multiplicative subgroup
\[\mathset{g^{4j} ~|~ 0\leq j\leq n-1}\subseteq \Z_{4n+1}.\]
It is also easily seen that
\[(g^j)^n=(g^{4\floorenv{j/4}+(j\bmod 4)})^n= (g^{4n})^{\floorenv{j/4}}g^{n(j\bmod 4)}
=g^{n(j\bmod 4)}.\] Since $n$ is odd, we get that an element
$a\in\Z_{4n+1}$, $a\neq 0$, is a quartic residue, i.e., $\chi(a)=1$,
if and only if $a^n\equiv 1\pmod{4n+1}$.

We are given that $6^n \not\equiv 1\pmod{4n+1}$, and thus
$1\neq \chi(6)=\chi(2)\chi(3)$. It follows that $\chi(2)=\chi(3)$.
We now have
\[\sum_{m\in M}\chi(m)=\chi(-1)+\chi(1)+\chi(2)+\chi(3)=\pm 2i \neq 0.\]
Therefore, $\sum_{s\in S}\chi(s)=0$, which is only possible if $\abs{S}=n$
is even, a contradiction.
\end{proof}

Up to dimension $250$ Theorem \ref{th:char4} rules out lattice tilings by
$(3,1,n)$-quasi-crosses for
\begin{align*}
n&= 3,7,9,13,15,25,27,39,45,49,57,67,69,73,79,87,93,99,105,127,135,\\
&\quad\ 153,165,175,177,183,189,193,205,207,213,219,249.
\end{align*}

We can also use higher order characters to obtain necessary conditions
for $(\kp,\km,n)$-quasi-cross to lattice tile $\R^n$ when $\kp+\km$ is
a prime. To that end we first need a simple lemma.

\begin{lemma}
\label{lem:rootsum}
Let $p$ be a prime and set $\omega=e^{2\pi i/p}$, $i=\sqrt{-1}$.  If
$a_0,\dots,a_{p-1}\in \Q$ are rational numbers such that
$\sum_{j=0}^{p-1}a_j \omega^j = 0$, then $a_0=a_1=\dots=a_{p-1}$.
\end{lemma}
\begin{proof}
Define the polynomial $a(x)=\sum_{j=0}^{p-1}a_j x^j\in\Q[x]$. It is
therefore given that $a(\omega)=0$, and hence all the conjugates of
$\omega$ relative to $\Q$ are also roots of $a(x)$. It is well-known
(see for example \cite{HarWri08}) that these are $\omega^j$ where
$\gcd(j,p)=1$. Since $p$ is a prime, we have that all of $\omega^j$,
$1\leq j\leq p-1$, are also roots of $a(x)$, i.e.,
\[(x-\omega^1)(x-\omega^2)\dots(x-\omega^{p-1}) ~|~ a(x).\]
However,
\[(x-\omega^1)(x-\omega^2)\dots(x-\omega^{p-1})=\frac{x^p-1}{x-1}
=1+x+x^2+\dots+x^{p-1}.\]
We now have
\[1+x+x^2+\dots+x^{p-1} ~|~ a(x)\]
while the degree of $a(x)$ is at most $p-1$, resulting in
\[a(x)=c(1+x+x^2+\dots+x^{p-1}),\]
for some constant $c\in\Q$.
\end{proof}

\begin{theorem}
Let $1\leq \km<\kp$ be positive integers such that $\kp+\km$ is an odd
prime. If the $(\kp,\km,n)$-quasi-cross lattice tiles $\R^n$, and
$n(\kp+\km)+1$ is a prime, then $\kp+\km ~|~ n$.
\end{theorem}
\begin{proof}
Denote $q=(\kp+\km)n+1$, and assume $\Z_q=(M,S)$ is a splitting with
$M=[-\km,\kp]^*$. Since $q$ is a prime $\Z_q$ is a field, and let $g$
be a primitive element in it.

We also denote $p=\kp+\km$, an odd prime, and let $\omega=e^{2\pi i
  /p}$ be a complex $p$-th root of unit. We define the character
$\chi:R_q\to\C$ as $\chi(g^j)=\omega^j$. Using the same argument as in
Theorem \ref{th:char4} we must have
\[\sum_{m\in M}\chi(m)=0 \qquad \text{or} \qquad \sum_{s\in S}\chi(s)=0.\]

We first check the characters of the elements in $M$. We have $1\in M$
and necessarily $\chi(1)=1=\omega^0$. We also have $-1\in M$, and
since $(-1)^2=1$, we get $\chi(-1)^2=\chi(1)=1$, but $p$ is an odd
prime and so $\chi(-1)=1=\omega^0$ also. If $\sum_{m\in M}\chi(m)=0$
then by Lemma \ref{lem:rootsum} each power of $\omega$ appears an
equal number of times, and since we have $p$ powers and $p$ summands,
each should appear exactly once. However, $\omega^0$ appears at least
twice, and so $\sum_{m\in M}\chi(m)\neq 0$.

It now follows that we must have $\sum_{s\in S}\chi(s)=0$, which again
by Lemma \ref{lem:rootsum} implies that $\kp+\km=p ~|~ n$, as claimed.
\end{proof}

\subsection{The Power Character}

An altogether different flavor of necessary conditions is obtained by
examining the power character which we now define: for any fixed
positive integer $r$, the function $\chi_r:R_q\to R_q$ defined by
$\chi_r(a)=a^r$, is a character we call the \emph{power
  character}. Unlike the previous section, we do not require $q$ to be
prime.

\begin{theorem}
\label{th:char}
There is no lattice tiling of $\R^n$ by $(4k-1,1,n)$-quasi-crosses for
all positive integers $k$ such that $kn\equiv 5,8\pmod{9}$.
\end{theorem}

\begin{proof}
Let us assume to the contrary that there exists a splitting
$\Z_{4kn+1}=(M,S)$ with $M=[-1,4k-1]^*$ and $\abs{S}=n$. Consider
the power character $\chi_2:R_{4kn+1}\to R_{4kn+1}$ defined by
$\chi_2(a)=a^2$. By Theorem \ref{th:charsum} it follows that
\[\parenv{\sum_{m\in M} m^2}\cdot\parenv{ \sum_{s\in S} s^2}
\equiv \sum_{i=1}^{4kn} i^2 \pmod{4kn+1}.\]

By a simple induction one can easily prove that
\[3 ~\left|~ (-1)^2+\sum_{i=1}^{4k-1}i^2\right.\]
for all $k\geq 1$. Thus, we can write
\begin{equation}
\label{eq:eq1}
3t \sum_{s\in S} s^2\equiv \frac{4k(4kn+1)(8kn+1)}{6} \pmod{4kn+1}
\end{equation}
for some integer $t$, where we used the well-known identity
\[\sum_{i=1}^{a} i^2 = \frac{a(a+1)(2a+1)}{6}.\]

We now note that $kn\equiv 5,8\pmod{9}$ implies
$4kn+1\equiv 3,6\pmod{9}$, and so $3$ is a zero divisor in $R_{4kn+1}$.
The LHS of \eqref{eq:eq1} is a multiple of $3$.
On the other hand, in the RHS of \eqref{eq:eq1},
$\frac{4kn}{2}$, $\frac{4kn+1}{3}$, and $8kn+1$, are all integers
leaving non-zero residue modulo $3$. This is a contradiction.
\end{proof}

\begin{theorem}
There is no lattice tiling of $\R^n$ by $(4k+2,1,n)$-quasi-crosses for
all positive integers $k$, and $n\equiv 3,7\pmod{8}$.
\end{theorem}

\begin{proof}
The proof is similar to that of Theorem \ref{th:char}.  Assume to the
contrary that there exists a splitting $\Z_{(4k+3)n+1}=(M,S)$ with
$M=[-1,4k+2]^*$ and $\abs{S}=n$. Consider the power character
$\chi_3:R_{(4k+3)n+1}\to R_{(4k+3)n+1}$ defined by $\chi_3(a)=a^3$. By
Theorem \ref{th:charsum} it follows that
\[\parenv{\sum_{m\in M} m^3}\cdot\parenv{ \sum_{s\in S} s^3}
\equiv \sum_{i=1}^{(4k+3)n} i^3 \pmod{(4k+3)n+1}.\]

By a simple induction one can easily prove that
\[8 ~\left|~ (-1)^3+\sum_{i=1}^{4k+2}i^3\right.\]
for all $k\geq 1$. Thus, we can write
\begin{equation}
\label{eq:eq2}
8t \sum_{s\in S} s^3\equiv \frac{((4k+3)n)^2((4k+3)n+1)^2}{4} \pmod{(4k+3)n+1}
\end{equation}
for some integer $t$, where we used the identity
\[\sum_{i=1}^{a} i^3 = \frac{a^2(a+1)^2}{4}.\]

At this point we note that $n\equiv 3,7\pmod{8}$ implies
$(4k+3)n+1\equiv 2,6\pmod{8}$, and so $2$ is a zero divisor in $R_{(4k+3)n+1}$.
The LHS of \eqref{eq:eq2} is a multiple of $2$.
On the other hand, the RHS of \eqref{eq:eq2} is odd since both
$((4k+3)n)^2$, and $\frac{((4k+3)n+1)^2}{4}$, are odd integers.
This is a contradiction.
\end{proof}

More elaborate results may be reached by using other power characters.
We turn to show a more general result using power characters.

\begin{theorem}
\label{th:vandermonde}
Let $\Z_q=(M,S)$ be a splitting, $n=\abs{S}<q-1$. If $q$ is a prime, then
\[\sum_{m\in M} m^i \equiv 0 \pmod{q}\]
for some $1\leq i\leq n$.
\end{theorem}
\begin{proof}
For every $1\leq i\leq n$ we consider the power character $\chi_i: R_q
\rightarrow R_q$ defined by $\chi_i(a)=a^i$. By Theorem \ref{th:charsum} we
therefore have
\begin{equation}
\label{eq:3}
\parenv{\sum_{m\in M} m^i}\cdot\parenv{ \sum_{s\in S}s^i} \equiv
\sum_{j=1}^{q-1} j^i \pmod{q}
\end{equation}
for all $1\leq i\leq n$.

If $q$ is a prime then $\Z_q$ is a field, its multiplicative group is cyclic,
and so let $g\in \Z_q$ be a primitive element in $\Z_q$. We can then write
\[\sum_{j=1}^{q-1} j^i \equiv \sum_{j=0}^{q-2} g^{ij} \equiv
\frac{g^{i(q-1)}-1}{g^i-1}\equiv 0 \pmod{q}\]
since $g^i\not\equiv 0 \pmod{q}$ for all $1\leq i\leq n<q-1$.

Since $\Z_q$ is a field, it now follows from \eqref{eq:3}, that for all
$1\leq i\leq n$ we have
\[\sum_{s\in S}s^i \equiv 0 \pmod{q} \qquad\text{or}\qquad 
\sum_{m\in M}m^i \equiv 0 \pmod{q}.\]
Assume to the contrary that for all $1\leq i\leq n$ we have
\[\sum_{m\in M} m^i \not\equiv 0 \pmod{q}.\]
If we define the matrix
\[
V=\begin{pmatrix}
s_1^1 & s_1^2 & \dots & s_1^{n} \\
s_2^1 & s_2^2 & \dots & s_2^{n} \\
\vdots & \vdots & \ddots & \vdots \\
s_{n}^1 & s_{n}^2 & \dots & s_{n}^{n}
\end{pmatrix}
\]
then it follows that
\[(1,1,\dots,1) V
\equiv (0,0,\dots,0) \pmod{q}\]
and so
\[\det(V)\equiv 0 \pmod{q}.\]
However, $V$ is clearly a Vandermonde matrix, and the elements of
$S$ are distinct, which implies
\[\det(V)=\prod_{j<j'} (s_j-s_{j'})\not\equiv 0 \pmod{q},\]
a contradiction.
\end{proof}

Up to dimension $250$ Theorem \ref{th:vandermonde} rules out lattice
tilings by $(3,1,n)$-quasi-crosses for a total of $59$ cases.

\subsection{Unique Representation}

By carefully examining the way specific elements of the split group
are represented we may sometimes reach a contradiction to the unique
representation of the group elements required by the splitting. The
following few results illustrate this method.

\begin{lemma}
\label{lem:kpfact}
If an integer $d$ divides $n(\kp+\km)+1$, $\gcd(d,\kp\#)=1$, and
$n<d<n(\kp+\km)+1$, then the $(\kp,\km,n)$-quasi-cross does not
lattice tile $\R^n$.
\end{lemma}
\begin{proof}
Denote $q=n(\kp+\km)+1$.  Assume to the contrary there is a splitting
$\Z_q=(M,S)$ with $M=[-\km,\kp]^*$. We note that $d$ is a zero divisor
in $\Z_q$ but not zero itself. According to the splitting, there is a
unique representation $d\equiv ms \pmod{q}$ with $m\in M$ and $s\in
S$. Since $\gcd(d,\kp\#)=1$ it follows that $\gcd(d,m)=1$ and
therefore $d~|~s$. Denote, then, $s=ds'$.

Since $d>n$ we have
\[\frac{q}{d}=\frac{n(\kp+\km)+1}{d}\leq \kp+\km.\]
Thus, there exist $m_1,m_2\in M$, $m_2\leq\km$, such that
\[m_1+m_2=\frac{q}{d}.\]
Then,
\[m_1s + m_2s = \frac{q}{d} s = q s' \equiv 0 \pmod{q},\]
and so
\[m_1 s \equiv -m_2 s \pmod{q}.\]
Since $m_1,-m_2\in M$ we have a contradiction to the splitting.
\end{proof}

The previous lemma gives rise to the following theorem.

\begin{theorem}
\label{th:kpprim}
For any $1 < r < \kp+\km$,
the $(\kp,\km,n)$-quasi-cross does not lattice tile $\R^n$ when
\[(\kp+\km)n+1\equiv ru \pmod{r\cdot\kp\#}\]
for all integers $u$ such that $\gcd(u,\kp\#)=1$.
\end{theorem}

\begin{proof}
We first note that reducing the requirement on $n$ modulo $r$ gives
\[(\kp+\km)n+1 \equiv 0 \pmod{r}.\]
Thus, $\frac{(\kp+\km)n+1}{r}$ is an integer and
\[\frac{(\kp+\km)n+1}{r} \equiv u \pmod{\kp\#}.\]
We can now use Lemma \ref{lem:kpfact} with
$d=\frac{(\kp+\km)n+1}{r}>n$, and the claim follows.
\end{proof}

If we try to apply Theorem \ref{th:kpprim} to the case of
$(3,1,n)$-quasi-crosses by setting $r=3$ we get the exact same result
as Theorem \ref{th:char}, i.e., no lattice tiling when
$n\equiv 5,8\pmod{9}$. We do, however, get new results for larger
quasi-crosses as the following example shows.
\begin{corollary}
\label{cor:32}
Both the $(3,2,n)$-quasi-cross and the $(4,1,n)$-quasi-cross do not
lattice tile $\R^n$ when
\begin{enumerate}
\item $n\equiv 5,9\pmod{12}$, or 
\item $n\equiv 4,10\pmod{18}$, or 
\item $n\equiv 15,23\pmod{24}$.
\end{enumerate}
\end{corollary}
\begin{proof}
We use Theorem \ref{th:kpprim} with $r=2,3,4$ for the three cases
respectively.
\end{proof}

\subsection{Recursion}

Recursion is also a powerful tool for formulating necessary conditions
for tilings. We present a simple recursion which may be used in
several ways to rule out lattice tilings.

\begin{theorem}
\label{th:div}
If there is a splitting $\Z_q=(M,S)$, with $M=[-\km,\kp]^*$, and
some positive integer $d~|~q$, $\gcd(d,\kp\#)=1$, then
\[(\kp+\km)d ~|~ q-d,\]
and there is a splitting $\Z_{q/d}=(M,S')$.
\end{theorem}
\begin{proof}
Let us consider the subgroup of $\Z_q$ defined by
\[H=d\Z\cap\Z_q=\mathset{0,d,2d,\dots,\parenv{\frac{q}{d}-1}d}.\]
Each element $id\in H$, $1\leq i\leq q/d-1$, has a unique representation
as
\begin{equation}
\label{eq:2}
id\equiv ms \pmod{q}
\end{equation}
with $m\in M$ and $s\in S$. Since $d$ is a zero divisor in $\Z_q$, and
$\gcd(d,\kp\#)=1$, it follows that $\gcd(d,m)=1$ and $d~|~s$. Denote
$s=ds'$ and reduce \eqref{eq:2} modulo $q/d$ to get
\[ i \equiv ms' \pmod{\frac{q}{d}}.\]
Define $S'=\mathset{s' ~|~ ds' \in S}$. Since every element of $H$ has
a unique factorization as in \eqref{eq:2}, it follows that
$\Z_{q/d}=(M,S')$ is indeed a splitting. Furthermore, the size of $S'$,
\[\abs{S'}=\frac{\abs{\Z_{q/d}-1}}{\abs{M}}=\frac{q-d}{(\kp+\km)d}\]
must be an integer.
\end{proof}

The following two corollaries follow immediately from Theorem \ref{th:div}:
The first is in fact a recursive construction, while the second may be
used to prove non-existence of lattice tilings.

\begin{corollary}
\label{cor:construct}
If the $(\kp,\km,n)$-quasi-cross lattice tiles $\R^n$, and for some
positive integer $d~|~(\kp+\km)n+1$ we have $\gcd(d,\kp\#)=1$, then
the $(\kp,\km,n')$-quasi-cross lattice tiles $\R^{n'}$,
$n'=\frac{(\kp+\km)n+1-d}{(\kp+\km)d}$.
\end{corollary}

\begin{corollary}
\label{cor:nonexist}
If there exists a positive integer $d~|~(\kp+\km)n+1$,
$\gcd(d,\kp\#)=1$, but $\frac{(\kp+\km)n+1-d}{(\kp+\km)d}$ is not an
integer, then the $(\kp,\km,n)$-quasi-cross does not lattice tile
$\R^n$.
\end{corollary}

We can turn Corollary \ref{cor:nonexist} into a more convenient form
of non-existence result in the following theorem.

\begin{theorem}
Let $p>\kp$ be a prime, $p\not\equiv 1\pmod{\kp+\km}$, and $p\neq \kp+\km$.
Then the $(\kp,\km,n)$-quasi-cross does not lattice tile $\R^n$ for
$n\equiv -(\kp+\km)^{-1} \pmod{p}$, where $(\kp+\km)^{-1}$ is the multiplicative
inverse of $\kp+\km$ in $\Z_p$.
\end{theorem}

\begin{proof}
We start by noting that $1\leq \km<\kp<p$ and $p\neq \kp+\km$ which
means $p \nmid \kp+\km$ and so $\kp+\km$ has a multiplicative
inverse in $\Z_p$. If
\[n\equiv -(\kp+\km)^{-1} \pmod{p}\]
then
\[(\kp+\km)n+1\equiv 0 \pmod{p}.\]
Thus, $p~|~(\kp+\km)n+1$.
However,
\[p\not\equiv 1\pmod{\kp+\km}\]
implies
\[(\kp+\km)n+1-p\not\equiv 0 \pmod{(\kp+\km)p}.\]
Since $p>\kp$ we must have $\gcd(p,\kp\#)=1$. We now use Corollary
\ref{cor:nonexist} with $d=p$.
\end{proof}

Even though Corollary \ref{cor:construct} was phrased as a recursive
construction, it can also be used to prove the non-existence of a
lattice tiling, as shown in the following theorem.

\begin{theorem}
Let $p$ be a prime, $p\equiv 1\pmod{\kp+\km}$. If the
$(\kp,\km,n)$-quasi-cross does not lattice tile $\R^n$, then the
$(\kp,\km,n')$-quasi-cross does not lattice tile $\R^{n'}$,
\[n'=\frac{\parenv{(\kp+\km)n+1}p^i-1}{\kp+\km},\]
for all positive integers $i$.
\end{theorem}
\begin{proof}
Assume to the contrary there is a lattice tiling of $\R^{n'}$ by
$(\kp,\km,n')$-quasi-crosses, where $n'=pn+\frac{p-1}{\kp+\km}$. We note that
$p~|~(\kp+\km)n'+1$, and that $p>\kp$ and so $\gcd(p,\kp\#)=1$. We now use
Corollary \ref{cor:construct} and get that there must be a lattice tiling
of $\R^{n''}$ by $(\kp,\km,n'')$-quasi-crosses, where
\[n''=\frac{(\kp+\km)n'+1-p}{(\kp+\km)p}=n,\]
a contradiction. Thus, there is not lattice tiling of $\R^{n'}$ by
$(\kp,\km,n')$-quasi-crosses. Repeating this argument $i$ times, for
any positive integer $i$, completes the proof.
\end{proof}

\subsection{Accounting for Zero Divisors}

The final approach we discuss is that of accounting for the way zero
divisors of the split Abelian group are represented, resulting in a
strong non-existence result.

\begin{theorem}
\label{th:psquare}
Let $p$ be a prime, and let $\kp$ and $\km$ be non-negative integers
such that $\km\leq \kp$ and $p\leq \kp < p^2$. Then the
$(\kp,\km,n)$-quasi-cross does not lattice tile $\R^n$ when
$(\kp+\km)n+1\equiv 0 \pmod{p^2}$ unless
\[n=\frac{p-1}{(\kp\bmod p)+(\km\bmod p)}.\]
\end{theorem}

\begin{proof}
Denote $q=(\kp+\km)n+1$, and assume to the contrary that $\Z_q=(M,S)$
is a splitting with $M=[-\km,\kp]^*$ and $q\equiv 0\pmod{p^2}$. Let us
consider the way the elements of
\[H=\frac{q}{p}\Z_p\setminus\mathset{0}=\mathset{\left. i\frac{q}{p} ~\right|~ 1\leq i\leq p-1}\]
are represented under this splitting.

We start by noting that $S\cap H=\emptyset$, for if some
$ip/q\in S$ then $p\cdot(ip/q)\equiv 0 \pmod{q}$ together with $p\in M$
contradict the splitting. We also note that all the elements of $H$ are
multiples of $p$, which is a zero divisor in $\Z_q$. Hence, every element of
$H$ is uniquely represented as $ms$, $m\in M$, $s\in S$, where
$m$ is a multiple of $p$. It follows that the number of multiples of $p$
in $M$ times the size of $S$ equals the size of $H$, i.e.,
\[\parenv{\floorenv{\frac{\kp}{p}}+\floorenv{\frac{\km}{p}}}n
= \frac{q}{p}-1 = \frac{(\kp+\km)n+1-p}{p}.\]
Hence, there is no splitting unless
\[n=\frac{p-1}{(\kp\bmod p)+(\km\bmod p)}.\]
\end{proof}

Theorem \ref{th:psquare} results in the following two corollaries.
\begin{corollary}
The $(3,1,n)$-quasi-cross does not lattice tile $\R^n$ when
$n\equiv 2\pmod{9}$, $n>2$.
\end{corollary}
\begin{proof}
Apply Theorem \ref{th:psquare} with $p=3$.
\end{proof}

\begin{corollary}
\label{cor:psq}
The $(3,2,n)$-quasi-cross does not lattice tile $\R^n$ when
\begin{enumerate}
\item
  $n\equiv 3 \pmod{4}$, or
\item
  $n\equiv 7 \pmod{9}$.
\end{enumerate}
\end{corollary}
\begin{proof}
Apply Theorem \ref{th:psquare} with $p=2,3$ respectively.
\end{proof}


\section{Conclusion}
\label{sec:conc}

In this work we showed, using a variety of techniques, several
necessary conditions for a quasi-cross of a given size to lattice tile
$\R^n$. Some of the results apply to general quasi-crosses, while
others are specific to quasi-crosses of small size. To conclude we
shall aggregate the results for the smallest unclassified
cases of the $(3,1,n)$-quasi-cross and the $(3,2,n)$-quasi-cross.

For the first shape, the $(3,1,n)$-quasi-cross, we recall there exists
a construction of lattice tilings from \cite{Sch12} for dimensions
$n=(5^i-1)/4$, $i\geq 1$. In addition, certain primes were shown in
\cite{YarKloBos13} to induce lattice tilings, as well as a recursive
construction, though a closed analytic form for the dimension appears
to be hard to obtain. Using a computer to verify the requirements for
the construction from \cite{YarKloBos13}, for $n\leq 250$ we also have
lattice tilings of $\R^n$ by $(3,1,n)$-quasi-crosses for dimensions
\[n=37,43,97,102,115,139,163,169,186,199,216.\]

On the other hand, combining the non-existence results with a nice
analytic form we achieved the following:

\begin{corollary}
If the $(3,1,n)$-quasi-cross lattice tiles $\R^n$ then $n\not\equiv
2\pmod{3}$.
\end{corollary}
\begin{proof}
The case of $n\equiv 5,8\pmod{9}$ is ruled out by Theorem
\ref{th:char}.  The case of $n\equiv 2\pmod{9}$, $n>2$, is ruled out
by Theorem \ref{th:psquare}. Finally, the case of $n=2$ is ruled out
by Theorem \ref{th:geometry}.
\end{proof}

However, especially for the $(3,1,n)$-quasi-cross, numerous other
non-existence results lacking a nice analytic form ensue from the
previous section. Aggregating the entire set of necessary conditions,
for $n\leq 250$, apart from the dimensions mentioned above allowing a
lattice tiling, no other lattice tiling of $\R^n$ by
$(3,1,n)$-quasi-crosses exists except perhaps in the remaining
unclassified cases of
\[n=22,24,60,111,114,121,144,220,234,235.\]

For the second shape, the $(3,2,n)$-quasi-cross, no lattice tiling is
known except for the trivial tiling of $\R^1$. The combined
non-existence results we obtained in this work, with a nice analytic
form, are much stronger in this case:

\begin{corollary}
If the $(3,2,n)$-quasi-cross lattice tiles $\R^n$ then $n\equiv
1,13\pmod{36}$.
\end{corollary}
\begin{proof}
This is a simple combination of Theorem \ref{th:kmo} stating
$5n+1\equiv 0\pmod{3}$, of Corollary \ref{cor:32} stating
$n\not\equiv 4,10\pmod{18}$, and of Corollary \ref{cor:psq}.
\end{proof}

Aggregating this result with the other recursive necessary conditions,
for $2\leq n\leq 250$, no lattice tiling of $\R^n$ by
$(3,2,n)$-quasi-crosses exists except perhaps in the remaining
unclassified cases of
\[n=13,37,49,73,85,121,145,157,181,217,229.\]

\bibliographystyle{elsarticle-num}
\bibliography{allbib}

\end{document}